\documentclass{scrartcl}

\usepackage{graphicx}     
\usepackage{amsfonts, amssymb, latexsym, mathtools}
\usepackage{amsmath}
\usepackage{amsthm}
\usepackage{graphicx}
\usepackage{epstopdf}
\usepackage{nicefrac}

\usepackage{color}

\newcommand{\eBox}{$\hfill\square$}

\newtheorem{defn}{Definition}

\newtheorem{prop}[defn]{Proposition}

\newtheorem{rem}[defn]{Remark}
\newtheorem{lemma}[defn]{Lemma}

\newtheorem{theorem}[defn]{Theorem}

\begin{document}

\title{A Dissipativity Characterization of Velocity Turnpikes in Optimal Control Problems for Mechanical Systems
}

\author{Timm Faulwasser\footnote{T.~Faulwasser,
			Department of Electrical Engineering and Information Technology, TU Dortmund University, Germany, e-mail: timm.faulwasser@ieee.org}
\and
Kathrin Fla\ss{}kamp\footnote{K. Fla{\ss}kamp, 
              Modeling and Simulation of Technical Systems, Systems Engineering, Saarland University, Germany, e-mail: kathrin.flasskamp@uni-saarland.de}
\and
Sina Ober-Bl\"{o}baum\footnote{S. Ober-Bl\"obaum,
             Department of Engineering Science,
             University of Oxford, United Kingdom,
             email: sina.ober-blobaum@eng.ox.ac.uk}
\and
Karl Worthmann\footnote{K. Worthmann,
          Institut f\"{u}r Mathematik,
         Technische Universit\"{a}t Ilmenau, Germany,
        email: karl.worthmann@tu-ilmenau.de}}

\maketitle

\begin{abstract}                
Turnpikes have recently gained significant research interest in optimal control, since they allow for pivotal insights into the structure of solutions to optimal control problems.
So far, mainly steady state solutions which serve as optimal operation points, are studied.
This is in contrast to time-varying turnpikes, which are in the focus of this paper.
More concretely, we analyze symmetry-induced velocity turnpikes, i.e.\ controlled relative equilibria, called trim primitives, which are optimal operation points regarding the given cost criterion.
We characterize velocity turnpikes by means of dissipativity inequalities.
Moreover, we study the equivalence between optimal control problems and steady-state problems via the corresponding necessary optimality conditions.
An academic example is given for illustration.
\end{abstract}



\section{Introduction}
Optimal control concepts are of key interest in the planning and computation of reference motions for mechanical systems. At the same time the inherent structure of mechanical systems implies specific properties.
In the classical work of \cite{Dubin1957}, later extended by \cite{Reeds90}, the trajectory planning problem for a car is solved geometrically, i.e.\ by concatenating straight lines and arcs of circles.
This approach shows two key points: (a) the existence of motions in mechanical systems of particularly simple shape (lines and arcs of circles), and (b) their concatenation to entire solution trajectories.
Conceptually, this has been formalized in \cite{FrDaFe05} by defining motion primitives as building blocks of trajectories and a proposed graph-based planning procedure to obtain sequences.
Among the building blocks, trim primitives, which are generated by the inherent system symmetry, are of particular interest.
Motion planning via trim primitives has gained recent interest in the trajectory design for autonomous driving \cite{Frazzoli2016, Kobilarov2017}.

Optimization is used in the planning procedure of \cite{FrDaFe05} and of the works based on the approach, e.g.~\cite{Ko08,FOK12}. However, one central question has not been addressed, yet. Namely, when is it optimal for the mechanical system to move on trim primitives?
In this paper, we address this question leveraging turnpike theory.

 Turnpikes are a classical concept in optimal control approaches in economics. While first observations can be traced back to \cite{vonNeumann38}, the notion as such has been coined by \cite{Dorfman58}, see also \cite{Mckenzie76,Carlson91}. In essence, the turnpikes phenomenon is a similarity property of optimal control problems parametric in the initial condition and the horizon length, i.e. for varying initial conditions and horizon length the time the optimal lifts spend close to a specific steady state grows with increasing horizon. In its easiest form, the turnpike is a the steady-state of the optimality system \cite{Trelat15a,kit:zanon18a}, while there have also been extensions to time-varying cases \cite{Gruene18b}. 

By now, it is well understood  that a dissipativity notion of Optimal Control Problems (OCPs), which was originally developed in context of so-called economic MPC---see \cite{kit:faulwasser18c} for a recent overview---plays a key role in analyzing turnpike properties, see \cite{epfl:faulwasser15h,Gruene16a}. Moreover there also exists a close relation between dissipativity, stability and reachability in infinite-horizon OCPs \cite{tudo:faulwasser20a}.

The present paper considers a specific class of OCPs arising for mechanical systems. We investigate the link between the concept of velocity turnpikes, which we recently proposed in \cite{kit:faulwasser19b_2}, and dissipativity properties of the underlying OCP. The core challenge of velocity turnpikes is that in contrast to the classical steady-state concept, the turnpike corresponds to a partial steady state where positions are required to be stationary. Specifically, we show that a suitable dissipativity notion of OCPs allows certifying velocity turnpikes and we characterize the reduced dynamics of the optimality system, which correspond to the velocity turnpike. 

The remainder of this paper is organized as follows: Since we bring together concepts from two fields of research --mechanical systems with symmetries and turnpike theory in optimal control-- we give basic definitions in Section~\ref{sec:preliminaries}. In Section~\ref{sec:VelTurnpikes} we introduce velocity turnpikes and show their dissipativity properties. Then, we focus on the adjoints and give the relation between the OCP and a velocity steady state problem in Section~\ref{sec:Adjoints}. An illustrative example is shown in Section~\ref{sec:Example}, before we close by giving an outlook to possible generalizations of our finding in future work in Section~\ref{sec:Conclusion}.

\section{Preliminaries}\label{sec:preliminaries}

\subsection{Mechanics and Symmetry}\label{sec:MechSymmetry}

The dynamics of mechanical systems are often given by Euler-Lagrange equations
\begin{align}\label{NotationSystemDynamics}
		\frac{\mathrm{d}}{\mathrm{d}t} \frac{\partial L}{\partial \dot{q}} - \frac{\partial L}{\partial q} & = f_L(q,\dot{q},u)
\end{align}
with real-valued Lagrangian $L$ and mechanical forces $f_L$.
Let $Q$ denote the $\frac{n}{2}$-dimensional smooth manifold ($\frac{n}{2}\in\mathbb{N}$) of configurations $q$, such that the
tangent bundle $TQ$ forms the $n$-dimensional state space.
The external controls are denoted by $u \in \mathbb{R}^m$. 
Assuming regularity of the Lagrangian, the second-order Euler-Lagrange equations can be reformulated as a system of first-order Ordinary Differential Equations (ODEs) in the form 
\[\dot{x} = f(x,u)\]
where $x = (q,\dot{q}) = (q,{v}) \in T_q Q$ denotes the full state, which is contained in the tangent space at~$q$. %
Then, the solution~$x(t) = \phi_u(t;x_0)$ to the Euler-Lagrange Eq.~\eqref{NotationSystemDynamics} for initial condition $x_0$ and $u \in L^\infty([0,T],\mathbb{R}^m)$ is given by the forced Lagrangian flow $\phi_u: [0,T] \times TQ \to TQ$.

In this paper, we consider mechanical systems which possess Lie group symmetries.
In general, a Lie group is a group $(\mathcal{G},\circ)$, which is also a smooth manifold, for which the group operations $(g,h) \mapsto g \circ h$ and $g \mapsto g^{-1}$ are smooth. If, in addition, a smooth manifold~$M$ is given, we call a map $\Psi: \mathcal{G} \times M \rightarrow M$ a left-action of $\mathcal{G}$ on $M$ if and only if the following properties hold:
\begin{itemize}
	\item $\Psi(e,x) = x$ for all $x \in M$ where $e$ denotes the neutral element of $(\mathcal{G},\circ)$,
	\item $\Psi(g,\Psi(h,x)) = \Psi(g \circ h,x)$ for all $g, h \in \mathcal{G}$ and $x \in M$.
\end{itemize}

\begin{defn}[Symmetry Group]\label{def:symmetryGroup}
	Let the configuration manifold $Q$ be a smooth manifold, $(\mathcal{G},\circ)$ a Lie-group, and $\Psi$ a left-action of $\mathcal{G}$ on $Q$. Further, let $\Psi^{TQ}: \mathcal{G} \times TQ \to TQ$ be the lift of $\Psi$ to $TQ$. Then, we call the triple $(\mathcal{G},Q,\Psi^{TQ})$ a \emph{symmetry group} of the system \eqref{NotationSystemDynamics} if the property
\begin{align}\label{NotationCommutativityFlowSymmetry}
	\phi_u(t;\Psi^{TQ}(g,x_0)) = \Psi^{TQ}(g,\phi_u(t;x_0)) \quad\forall\, t \in [0,T]
\end{align}
holds for all $(g,x_0,u) \in \mathcal{G} \times TQ \times L^\infty([0,T],\mathbb{R}^m)$.
\eBox
\end{defn}
Given a mechanical system with symmetry group, trajectories which are equivalent w.r.t.\ the symmetry action can be identified. A \emph{motion primitive} denotes the equivalence class of all equivalent trajectories for a fixed $g\in\mathcal{G}$ and given control signal.

Moreover, the symmetry may lead to the existence of 
special trajectories, which we call \emph{trim primitives} (\emph{trims} for short).
\begin{defn}[Trim Primitive] \label{def:Trims}
Let $(\mathcal{G},Q,\Psi^{TQ})$ be a symmetry group in the sense of 
	Definition~\ref{def:symmetryGroup}.
	Then, a trajectory $\phi_u(\cdot;x_0)$, $u(t) \equiv \bar{u} = \text{const.}$, is called a \emph{(trim) primitive} if there exists a Lie algebra element $\xi \in\mathfrak{g}$ such that
	\begin{equation}\label{eq:defTrim} 
		\phi_u(t;x_0) = \Psi^{TQ}(\exp(\xi t),x_0) \quad\forall\, t \geq 0.  
	\end{equation}
	\eBox
\end{defn}
For a formal definition of Lie algebras we refer to~\cite{Bake12}.

In this paper, we will focus on mechanical systems for which the Lagrangian and the mechanical forces are configuration independent.
That is, we consider mechanical systems of the particular form
		\begin{equation}\label{NotationSecondOrderSystem}							\begin{aligned}
			\dot{q}(t) & = v(t) \\
			\dot{v}(t) & = f(v(t),u(t))
			\end{aligned}
		\end{equation}
Thus, the system is independent, i.e.\ symmetric w.r.t.\ translations in all configuration variables $q$.
The corresponding Lie group $\mathcal{G}$ is identical to the full configuration manifold and operates via vector addition, i.e.\ $\Psi(g,q) = q+g$ and $\Psi^{TQ}(g,x) =(q+g,v)$ . 
\begin{lemma} \label{lem:TrimsSimple}
Given a mechanical system of type \eqref{NotationSecondOrderSystem}, a trim can be characterized by the pair $(\bar{v},\bar{u})^\top$ satisfying the condition 
\begin{equation} \label{eq:TrimCondition}
f(\bar{v},\bar{u}) = 0.
\end{equation}
\end{lemma}
\begin{proof}	Let $(q_0,\bar{v})$ denote the initial value.
The corresponding solution for control $u(t)\equiv \bar{u}$ is  $q(t) = q_0 + \bar{v}t$ and $v(t) = v_0 = \bar{v}$. This can also be expressed via
		\begin{align} \label{eq:TrimSimple}
			\Psi^{TQ} \left(\exp(\xi t), \begin{pmatrix}
				q_0 \\ v_0
			\end{pmatrix} \right) & 
			= \begin{pmatrix}
				q_0 + \xi t \\ v_0
			\end{pmatrix} 
		\end{align}
with $\xi = \bar{v}$ according to Definition~\ref{def:Trims}.
\end{proof}

\subsection{Turnpikes in Optimal Control}
Let the stage cost $\ell: \mathbb{R}^{n} \times \mathbb{R}^{m} \to \mathbb{R}$ be  continuous and convex and let the closed sets $\mathbb{U} \subseteq \mathbb{R}^m$ and $\mathbb{X} \subseteq \mathbb{R}^n$ be given.
A general OCP is given as
\begin{align}
	\underset{u \in L^\infty([0,T],\mathbb{R}^{m})}{\text{minimize}} \quad & \int_{0}^{T} \ell(x(t),u(t))\,\mathrm{d}t \nonumber \\
	\text{subject to}& \label{eq:OCPclassic}\\
			\dot{{x}}(t) & = f({x}(t),{u}(t))\quad \forall\, t \in [0,T]\nonumber\\
		 		 {x}(0) &= {x}_0 \text{ and } {x}(T) = {x}_{T} \nonumber\\
		 u(t) &\in \mathbb{U} \text{ and } x(t) \in \mathbb{X} \quad \forall\, t \in [0,T] \nonumber
\end{align}
where the last three conditions refer to the system dynamics, the boundary conditions, and the control and state constraints.

%
%
%
\begin{defn}\label{def:controlledEqClassic}
A state~$x \in \mathbb{X}$ is called \emph{(controlled) equilibrium} if there exists $u \in \mathbb{R}^m$ such that $f(x,u) = 0$ holds. %
Based on this terminology, 	the pair $(x,u) \in \mathbb{R}^{n} \times \mathbb{R}^{m}$ is called an \emph{optimal steady state} if it holds that
\begin{align}\label{eq:SOPclassic}
	(x,u) = \operatorname{argmin} \{ \ell(x,u) | (x,u) \in \mathbb{X} \times\mathbb{U}, f(x,u) = 0 \}.
\end{align}
\eBox
\end{defn}

\begin{figure}[t]
\begin{center}
\includegraphics[scale=0.4]{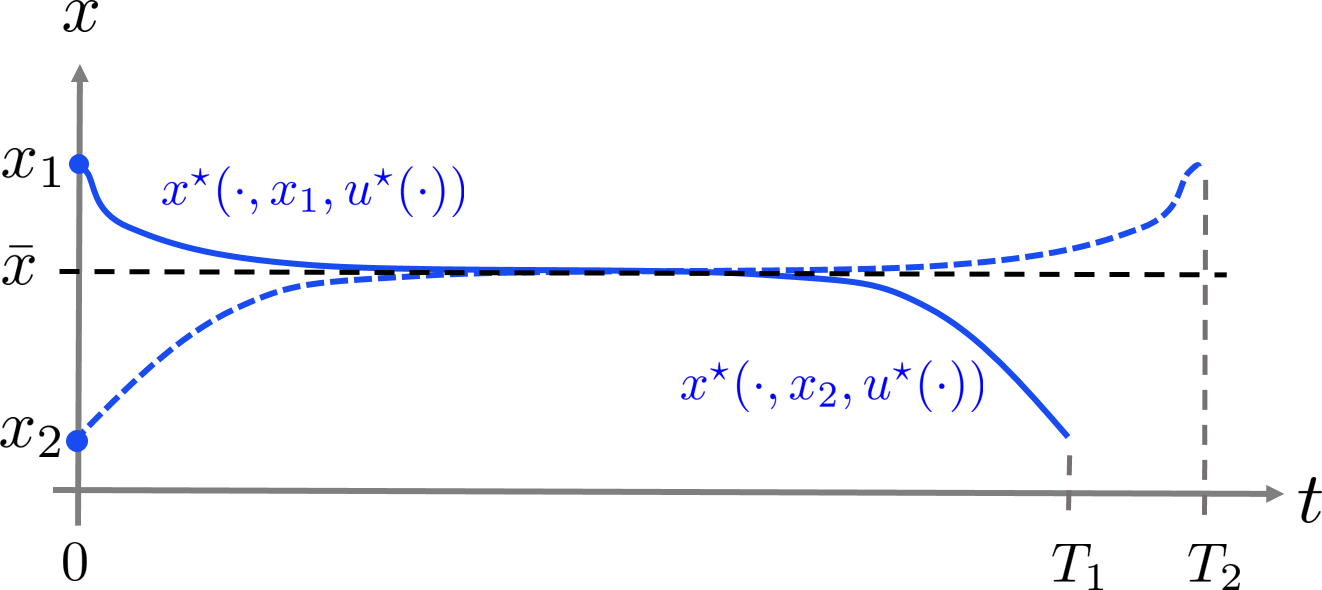}
\end{center}
\caption{Sketch of a classical steady-state turnpike \label{fig:turnpike}}
\end{figure}

Classically, \emph{turnpikes} are optimal steady states, i.e. solutions to \eqref{eq:SOPclassic}, see \cite{Mckenzie76,Carlson91}. As sketched in Figure \ref{fig:turnpike}, for different initial conditions $x_1, x_2$ and varying horizons $T_1, T_2$ the optimal solutions spend an increasing amount of time close to the turnpike $\bar x$, which turns out to be an optimal steady state. 
Only if the horizon is too short, it may be too costly to approach the respective steady state and thus, the turnpike phenomenon vanishes. We remark without further elaboration that there exist varying definitions of turnpike properties, see \cite{Stieler14a,Trelat15a} for so-called \textit{exponential} turnpikes, \cite{Gugat16} for \emph{integral} turnpikes, and \cite{epfl:faulwasser15h} for measure turnpikes. 
Turnpikes are also closely related to dissipativity properties of OCPs \cite{Gruene16a,epfl:faulwasser15h} and to stability properties of infinite-horizon OCPs \cite{tudo:faulwasser20a}. 


\section{Velocity Turnpikes and Dissipativity}\label{sec:VelTurnpikes}
We will now extend the concept of turnpikes to mechanical systems with symmetries. To this end, consider a mechanical system with invariances as defined in \eqref{NotationSecondOrderSystem}. In particular, the set of admissible states $x=(q,v)$  is $\mathbb{X} \subseteq TQ$, a subset of the tangent bundle.
We consider the OCP
\begin{align}
	\underset{u \in L^\infty([0,T],\mathbb{R}^{m})}{\text{minimize}} \quad & \int_{0}^{T} \ell(v(t),u(t))\,\mathrm{d}t \nonumber \\
	\text{subject to}& \label{eq:OCP}\\
		\dot{{q}}(t) & = {v}(t) \nonumber\\
			\dot{{v}}(t) & = f({v}(t),{u}(t))\quad \forall\, t \in [0,T]\nonumber\\
		 		 {q}(0) &= {q}_0,\, v(0)=v_0 \text{ and } {q}(T) = {q}_{T},\, v(t)=v_T, \nonumber\\
		 u(t) &\in \mathbb{U} \text{ and } ({q}(t),v(t)) \in \mathbb{X} \quad \forall\, t \in [0,T]. \nonumber
\end{align}%
Note that we now also assume the stage cost $\ell$ to be independent w.r.t.\ $q$.

For a controlled equilibrium, we necessarily have $v\equiv 0$ and, thus, $u$ such that $f(0,u)=0$ holds.
In the following, we are also interested in zeros of $f$ with non-zero velocity, i.e.~in trims (cf.~Definition~\ref{def:Trims}).

For the system class defined in Eq.~\eqref{NotationSecondOrderSystem}, a trim corresponds to an equilibrium relative to the dynamics in $v$, but not to the dynamics in $q$.
Thus, it has been introduced as a \emph{velocity steady state} in \cite{kit:faulwasser19b_2}.


\begin{defn}

	Let $(v,u)$ be a trim as characterized in Lemma~\ref{lem:TrimsSimple}. %
	The pair $(v,u)$ is called an \emph{optimal velocity steady state} if it holds that
	\begin{align}\label{eq:SOP}
		(v,u) = \operatorname{argmin}  \{ \ell(v,u) \,  |  \,(v,u) \in \mathbb{X}_V\times\mathbb{U}, f(v,u) = 0 \}.
	\end{align}
	where $\mathbb{X}_V := \{ v \, | \, \exists\ x: (x,v) \in \mathbb{X} \}$ is the projection of $\mathbb{X}$ on the $v$-component.
\eBox
\end{defn}

Note that in contrast to the classical definition of an optimal steady state (Definition~\ref{def:controlledEqClassic}), an optimal velocity steady state does not define the full state vector, but only the $v$-component.
We decide not to fix the initial configuration $q_0$ of the corresponding trim, since any other configuration $\tilde{q} = q_0 + v\cdot t$ for some $t \in \mathbb{R}$ would define the same trim.
This is due to the symmetry equivalence (cf. Section~\ref{sec:MechSymmetry}).


Next we recall a definition of a velocity turnpike property, where the turnpike as such is a trim, see~\cite{kit:faulwasser19b_2}. %
Similarly to \cite{Carlson91, epfl:faulwasser15h} consider 
\begin{equation} \label{eq:Theta}
\Theta_{T}(\varepsilon) = \left\{t \in [0,T]: \left\| (v^\star(t), u^\star(t)) - (\bar v, \bar u)\right\| > \varepsilon\right\},
\end{equation}
which is the set of time points for which the optimal velocity and input trajectory pairs is not inside an $\varepsilon$-ball of the steady-state pair $(\bar v, \bar u)$. Now we are ready to define a measure-based velocity turnpike property similar to \cite{epfl:faulwasser15h}. 
\begin{defn}[Velocity turnpike property]\label{def:turnpike}~\\ 
	The optimal solutions $(q^\star(\cdot), v^\star(\cdot), u^\star(\cdot))$ are said to have a \textnormal{velocity turnpike} with respect to $(\bar v, \bar u)$ %
	if there exists a function $\nu:[0,\infty)\to [0,\infty]$ such that, for all $(q_0, v_0) \in \mathbb{X}_0 \subseteq \mathbb{X}$ and all $T>0$, we have
\begin{equation}\label{eq:TP}
\mu\left[\Theta_{T}(\varepsilon)\right]  <\nu(\varepsilon)< \infty \quad \forall\: \varepsilon >0,
\end{equation}
where $\mu[\cdot]$ is the Lebesgue measure on the real line.\\
The  optimal solutions $(q^\star(\cdot), v^\star(\cdot), u^\star(\cdot))$ are said to have an \textnormal{exact velocity turnpike} if Condition \eqref{eq:TP} also holds for $\varepsilon = 0$, i.e., 
\begin{equation} \label{eq:exactTP}
\mu\left[\Theta_{T}(0)\right]  <  \nu(0) < \infty. ~ 
\end{equation}\eBox
\end{defn}
 %

 %
%
Next, we adopt the definition of dissipativity with respect to a steady state \cite{Angeli12a} for our setting. We refer to \cite{Willems07a, Moylan14a} for further details on dissipativity. 
Let $w: \mathbb{X}\times\mathbb{U}\to \mathbb{R}$  be given by
\begin{equation} \label{eq:def_w}
w(v,u) := \ell(v,u) - \ell(\bar v,\bar u),
\end{equation}
where $\ell$ is the stage cost in the OCP~\eqref{eq:OCP}. 
\begin{defn}[Dissipativity w.r.t.\ a velocity steady state]\label{def:diss} 
OCP \eqref{eq:OCP} is said to be \textit{dissipative with respect to $ (\bar v, \bar u)^\top $} if 
there exists a non-negative storage function\footnote{Note that the required properties of $S$ differ in different works: in \cite{epfl:faulwasser15h, Mueller14a} boundedness is assumed, while in \cite{Angeli12a} the storage $S$ can take real values instead of non-negative real values.} $S:\mathbb{X} \to \mathbb{R}^+_0$ such that for all $(q_0, v_0) \in \mathbb{X}_0 \subseteq \mathbb{X}$, all $T\geq 0$ and all optimal input $u^\star(\cdot) \in L^\infty([0,T],\mathbb{U})$ we have
\begin{subequations} \label{eq:DI}
	\begin{equation}\label{eq:diss}
		S(q_T, v_T) - S(q_0, v_0) \leq \int_{0}^{T} w(v^\star(t), u^\star(t))\,\mathrm{d}t,
	\end{equation}
	where $(q_T, v_T) = (q^\star(T, u^\star(\cdot)), v^\star(T, u^\star(\cdot)))$. 
	If, in addition, there exists a continuous, strictly increasing function $\alpha: [0,\infty) \rightarrow \mathbb{R}$ with $\alpha(0) = 0$ satisfying
	\begin{multline} \label{eq:str_diss}
		S(q_T, v_T) - S(q_0, v_0) \leq \int_{0}^{T} -\alpha\left(\left\|(v^\star(t), u^\star(t))-(\bar v, \bar u)\right\|\right) + w(v^\star(t), u^\star(t))\,\mathrm{d}t
	\end{multline}
\end{subequations}
then, OCP \eqref{eq:OCP} is said to be \textnormal{strictly dissipative with respect to $(\bar v, \bar u)^\top $}.
 \eBox
\end{defn}
\begin{lemma}[Optimality of velocity steady state]
Let system \eqref{NotationSecondOrderSystem} be strictly dissipative with respect to $(\bar v, \bar u)^\top $, then it is the unique globally optimal minimizer in
\begin{align}
\min_{u\in\mathbb{U},v\in\mathbb{X}_V} \ell(v,u) \quad  \text{s.t. } f(v,u) = 0. \label{eq:SOPuv}
 \end{align}
 \eBox
 \end{lemma}
Proof follows directly from \eqref{eq:str_diss} in differential form. 
\begin{prop}[Dissipativity $\Rightarrow$ velocity turnpike]\label{prop:turnpike} Consider OCP \eqref{eq:OCP} and fix $T_0 > 0$. Let $\mathbb{X}_0$ be defined as the set of all initial states $x_0 =(q_0, v_0)^\top \in \mathbb{X}$ such that there exists a control $u = u(x_0) \in L^\infty([0,T_0], \mathbb{U})$ with $v(T_0,v_0,u(\cdot)) = \bar v$. Suppose that 
\begin{itemize}
	\item the considered terminal state $x_T = (q_t, v_T)$ is such that there exist a control $u_T \in L^\infty([0,T_T], \mathbb{U})$  with
\[
v(T_T, \bar v, u_T(\cdot)) = v_T;
\]
\item and let system \eqref{NotationSecondOrderSystem} be strictly dissipative with respect to  $(\bar v, \bar u)^\top $.
\end{itemize}
Then OCP exhibits a velocity turnpike in the sense of Definition \ref{def:turnpike}.\eBox
 \end{prop}
\begin{proof}
We assume without loss of generality that $\ell(\bar v, \bar u) =0$ and that the horizon is $T \geq T_0 + T_T$. The strict dissipation inequality with bounded storage implies
\begin{equation*} 
-2\hat S +\int_{0}^{T} \alpha\left(\left\|(v(t), u(t))-(\bar v, \bar u)\right\|\right)\,\mathrm{d}t \leq  \int_{0}^{T} \ell(v(t), u(t))\,\mathrm{d}t,
\end{equation*} 
with $\hat S = \sup S(x)$.
The reachability assumptions imply that for any optimal solution the performance can be bounded from above by
\[
\int_{0}^{T} \ell(v^\star(t), u^\star(t))\,\mathrm{d}t \leq C.
\]
Moreover, we split the time horizon $[0,T]$ into $\Theta_\varepsilon$ and $[0,T]\setminus \Theta_\varepsilon$ and
have the following bound
\[
\mu[\Theta_\varepsilon]\alpha(\varepsilon)=\int_{\Theta_\varepsilon}\alpha\left(\varepsilon\right)\,\mathrm{d}t \leq \int_{0}^{T} \alpha\left(\left\|(v(t), u(t))-(\bar v, \bar u)\right\|\right)\,\mathrm{d}t.
\]
Combining the last three inequalities yields
\[
\mu[\Theta_\varepsilon] \leq \dfrac{2\hat S + C}{\alpha(\varepsilon)}.
\]
\end{proof}


\section{Relation of Optimality Conditions} \label{sec:Adjoints}

In this section we compare optimality conditions of the OCP~\eqref{eq:OCP} and the velocity steady state optimization problem \eqref{eq:SOPuv}. 
%
%
These derivations need to assume that there are no state or input constraints, respectively,and that the optimal trim is characterized by an interior point of the velocity and input constraints. 
First we derive the necessary optimality conditions for the OCP~\eqref{eq:OCP} based on Pontryagin's maximum principle (PMP) which yields the adjoint equations
\begin{subequations}
\begin{align}
\dot{\lambda}_q & = 0\label{eq:adjoint1}\\
\dot{\lambda}_v & = -\frac{\partial l}{\partial v} - \lambda_q - \frac{\partial f}{\partial v}^T \, \lambda_v\label{eq:adjoint2}
\end{align}
and the optimality condition
\begin{align}
0 & =  \frac{\partial l}{\partial u} + \frac{\partial f}{\partial u}^T \, \lambda_v\label{eq:optimal}
\end{align}
\end{subequations}
for time varying adjoint variables $\lambda_q,\lambda_v$. 
The scalar-valued multiplier for the cost function has been set to one w.l.o.g.\ since this multiplier being zero requires all other multipliers to be zero, too --a case which is excluded in the PMP.
Necessary optimality conditions for the velocity steady state optimization problem \eqref{eq:SOPuv} are
\begin{subequations}
\begin{align}
0 & = f(v,u)\\
0 & = \frac{\partial l}{\partial v} + \frac{\partial f}{\partial v}^T \, \lambda\label{eq:adjoint_ss}\\
0 & = \frac{\partial l}{\partial u}  + \frac{\partial f}{\partial u}^T \, \lambda\label{eq:adjoint_ssu}
\end{align}
\end{subequations}
with constant Lagrange multiplier $\lambda$.
Comparing both sets of necessary optimality conditions we can derive conditions on the adjoints under which solutions are the same for both problems (cf.\ Figure~\ref{fig:commutediagram}).
\begin{prop}\label{prop:adjoint}
If there exist time intervals $[t_1,t_2]$ with $0\leq t_1< t_2 \leq T$ on which $\lambda_v \equiv \lambda =$ constant and if $\lambda_q(t_i)=0$ for $t_i \in [0,T]$, necessary conditions of the optimal control problem \eqref{eq:OCP} on the time interval $[t_1,t_2]$ reduce to optimality conditions of the velocity steady state optimization problem \eqref{eq:SOPuv}.\eBox
\end{prop}
\begin{proof}
The condition $\lambda_v \equiv \lambda =$ constant ensures equivalence of the optimality conditions \eqref{eq:optimal} and \eqref{eq:adjoint_ssu} on $[t_1,t_2]$. $\lambda_q(t_i)=0$ for $t_i \in [0,T]$ together with \eqref{eq:adjoint1} yields $\lambda_q\equiv 0$ on $[0,T]$ such that \eqref{eq:adjoint2} reduces to \eqref{eq:adjoint_ss}.
\end{proof}
%
%
%

\begin{figure}
\centering
\hspace{-1cm}\includegraphics[height =.6\textwidth]{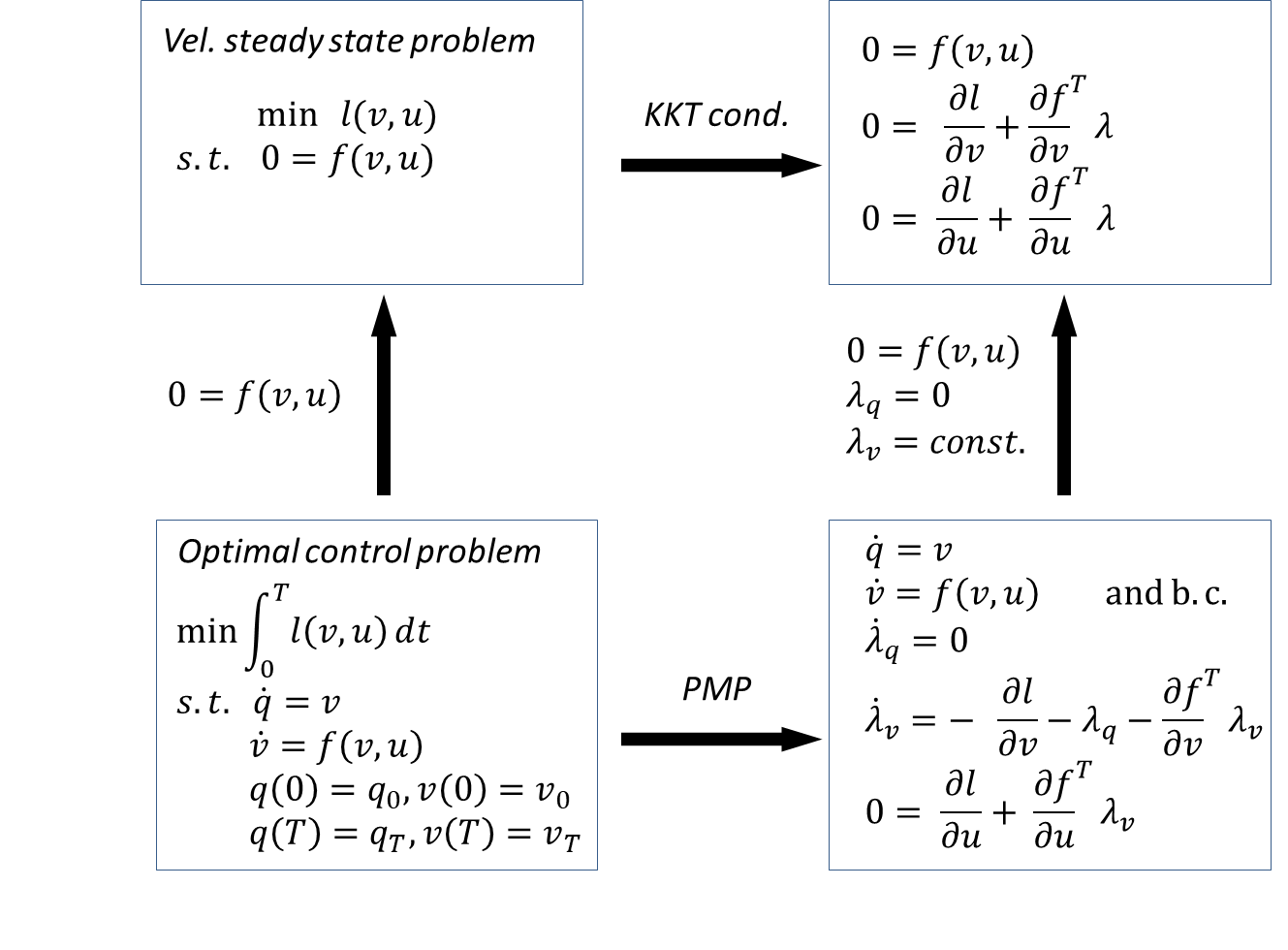}
\caption{Derivation of necessary optimality conditions (by Pontryagin's maximum principle (PMP) and Karush-Kuhn-Tucker (KKT) conditions) and imposing velocity steady state conditions commute.}
\label{fig:commutediagram}
\end{figure}

\begin{rem}

As a consequence of Proposition~\ref{prop:adjoint}, we see that on time intervals, where the dual parts of the optimality system coincide, then on these time intervals the optimal solutions will be at the velocity turnpike, which is specified by the optimal trim. 
In light of Proposition \ref{prop:turnpike}, if---for specific primal boundary conditions and provided the horizon is sufficiently long---such time intervals do not exist, then the optimal solutions still have to be close to the optimal trim solution of the steady state problem. 
Moreover, for regular optimal control problems, one expects that for general boundary conditions, which do not coincide with the turnpike,  the optimal solutions approach a neighborhood of the turnpike without reaching it exactly, see \cite{kit:faulwasser17a,kit:zanon18a} for the analysis of exact and non-exact steady-state turnpikes. Though a detailed analysis of exact variants of velocity turnpikes is beyond the scope of the present paper.\eBox

\end{rem} 
 

\section{Illustrative Example}\label{sec:Example}

We consider the second-order system $\ddot{x}(t) = u(t)$ written as a first-order ODE, i.e.
\begin{equation}\label{MotivationalExampleFirstOrderODE}
	\frac {\mathrm{d}}{\mathrm{d}t} \begin{pmatrix}
		q(t) \\ v(t)
	\end{pmatrix} = \begin{pmatrix}
		0 & 1 \\ 0 & 0
	\end{pmatrix} \begin{pmatrix}
		q(t) \\ v(t)
	\end{pmatrix} + \begin{pmatrix}
		0 \\ 1
	\end{pmatrix} u(t).
\end{equation}
Using the stage cost $\ell(v,u) := \frac 12 (\Vert v \Vert^2 $ $+ \Vert u \Vert^2)$ and imposing the boundary conditions
\begin{equation}\label{MotivationalExampleBoundaryConditions}
	\begin{pmatrix}
		q(0) \\ v(0)
	\end{pmatrix} = \begin{pmatrix}
		q_0 \\ v_0
	\end{pmatrix} \quad\text{ and }\quad \begin{pmatrix}
		q(T) \\ v(T)
	\end{pmatrix} = \begin{pmatrix}
		q_T \\ v_T
	\end{pmatrix},
\end{equation}
we get the OCP
\begin{align} 
	\underset{u \in \mathcal{L}^1([0,T],\mathbb{R})}{\text{minimize}}\quad & %
	\int_{0}^{T} \frac 12 \left( \Vert v(t) \Vert^{2} + \Vert u(t) \Vert^{2} \right) \mathrm{d}t \label{eq:example_OCP} \\
	\text{subject to}\quad & \eqref{MotivationalExampleFirstOrderODE} \text{ for almost all $t \in [0,T]$ and \eqref{MotivationalExampleBoundaryConditions}}. \nonumber
\end{align}

Since the system is invariant w.r.t.\ translations in $q$, %
any triple $(q,v,u)$ with $(q,0,0)$ is a velocity steady state satisfying $\ell(v,u) = 0$. %
Hence, the system is optimally operated at all of these steady states.

\begin{theorem}\label{ThmMotivationalExample}%
	For each optimization horizon~$T > 0$, %
	the optimal control problem~\eqref{eq:example_OCP} has a unique optimal solution $(q^\star, v^\star, u^\star): [0,T] \rightarrow \mathbb{R}^3$. %
	Moreover, the OCP~\eqref{eq:example_OCP} exhibits a hyperbolic velocity turnpike with respect to $(\bar v, \bar u) = (0,0)$, %
	i.e.\ for each bounded set~$K \subset \mathbb{R}^4$, there exists a positive constants $C$, $\bar\nu > 0$ such that, %
	for all initial conditions $(q_0\ q_T\ v_0\ v_T)^\top \in K$ and all $T > 0$, 
	we have
	\begin{equation}\label{eq:vTP}
		\left\| (v^\star(t)\ u^\star(t))^\top - (\bar{v}\ \bar{u})^\top \right\| \leq C / T
	\end{equation}
	for all $t \in [\bar\nu,T-\bar\nu]$. Furthermore, Inequality~\eqref{eq:vTP} also holds, %
	if the left hand side is replaced by $| \lambda_v(t) |$ where $\| (\lambda_q\ \lambda_v)^\top \|$ denotes the adjoint variables.\eBox
\end{theorem}
\begin{proof}
	Existence and uniqueness of an optimal solution can be shown analogously to \cite{kit:faulwasser19b_2}; note that the stage cost has not changed. %
	First, let us briefly recap some of the findings: Applying Pontryagin's Maximum Principle based on the Hamiltonian (OCP~\eqref{eq:example_OCP} is normal)
	\begin{equation*}
		\mathcal{H}(q,v,\lambda,u) = \frac{1}{2} \Big( v^{2} + u^{2} \Big) + \lambda_q v + \lambda_v u.
	\end{equation*}	
	yields the necessary optimality condition
	\begin{equation}
		u^\star(t) = -\lambda_v(t) \qquad\text{for almost all $t \in [0,T]$.} \label{eq:exampeLambdaV}
	\end{equation}
	Moreover, the solution of the state-adjoint system is given by %
	$(q^\star(t)\ v^\star(t)\ \lambda_q(t)\ \lambda_v(t))^\top = e^{At} (q_0\ v_0\ \lambda_q(0)\ \lambda_v(0)^\top$ with
	\begin{eqnarray}\nonumber
		e^{At} & = & \left( \begin{array}{rrrr}
			1 & \sinh(t) & \sinh(t)-t & 1 - \cosh(t) \\
			0 & \cosh(t) & \cosh(t)-1 & -\sinh(t) \\
			0 & 0 & 1 & 0 \\
			0 & -\sinh(t) & -\sinh(t) & \cosh(t)
		\end{array} \right)
	\end{eqnarray}
	where we used the functions $\cosh(t) = \nicefrac{1}{2} (e^t+e^{-t})$ and $\sinh(t) = \nicefrac{1}{2} (e^t-e^{-t})$ to simplify the resulting expression. %
	The initial value of the adjoint are given by
	\begin{align*}
		\lambda_q(0) & = \frac{\sinh(T)(q_T-q_0)+(1-\cosh(T)) (v_0+v_T)}{2 (\cosh(T)-1) - T \sinh(T)}, \\
		\lambda_v(0) & = \frac {\cosh(T) v_0 - v_T + (\cosh(T)-1) \lambda_q(0)}{\sinh(T)}.
	\end{align*}
	Note that $\lambda_q(t) = \lambda_q(0)$ holds for all $t \in [0,T]$ (in particular, $\lambda_q(T) = \lambda_q(0)$ holds). 

	In \cite[Proposition~8]{kit:faulwasser19b_2} it was shown for $v_0 = v_T = 0$, that the optimal velocity trajectory is given by
	\begin{align}
		v^\star(t) & = \left(\frac {\sinh(t) + \sinh(T-t) - \sinh(T)}{2(\cosh(T)-1) - T\sinh(T)}\right) \, (q_T-q_0). \label{eq:exampleZeroVelocity}
	\end{align}
	In the general case considered here, $v^\star(t)$ consists of the sum of its counterpart for $v_0 = v_T = 0$, %
	i.e.\ the right hand side of~\eqref{eq:exampleZeroVelocity}, the term
	\begin{equation}
		\frac {\sinh(T-t) v_0 + \sinh(t) v_T}{\sinh(T)}, \label{eq:exampleIncomingLeavingArc}
	\end{equation}
	which represents the (exponentially) decreasing influence of the initial velocity and the (exponentially) increasing impact of the terminal velocity, and
	\begin{align}
		& \frac {(1-\cosh(T)) (\sinh(T-t)-\sinh(T)+\sinh(t)) (v_0+v_T) }{\sinh(T) \cdot (2 (\cosh(T)-1)-T \sinh(T))}. \nonumber
	\end{align}
	Combining the last expression with the right hand side of~\eqref{eq:exampleZeroVelocity} yields the term
	\begin{equation}
		\left[ \left( \frac {1-\cosh(T)}{\sinh(T)} \right) (v_0+v_T) + (q_T-q_0) \right] \label{eq:exampleFactor}
	\end{equation}		
	multiplied with the factor
	\begin{align}
		\frac {\sinh(T-t)-\sinh(T)+\sinh(t)}{2 (\cosh(T)-1)-T \sinh(T)}. \nonumber
	\end{align}
	Then, following the same line of reasoning as presented in~\cite{kit:faulwasser19b_2} yields %
	that this factor is uniformly bounded by $\tilde{c}/T$ with constant $\tilde{c} := 3/2$ on the time interval $[0,T]$.
	Since the factor $(1-\cosh(T)) / \sinh(T)$ is monotonically increasing in the optimization horizon~$T$ %
	with being equal to zero for $T = 0$ and converging to one for $T \rightarrow \infty$, %
	these two summands are uniformly bounded by $c/T$ with $c := \tilde{c} (|v_0+v_T|+|q_T-q_0|)$. %
	
	Since we may rewrite the quotient $\sinh(T/2) / \sinh(T)$ as $(2 \cosh(T/2))^{-1}$, %
	the third summand~\eqref{eq:exampleIncomingLeavingArc} in the representation of~$v^\star(t)$,
	which essentially represents the incoming and the arrival arc, is exponentially decaying with increasing distance to the boundaries. %

	In conclusion, the optimal solutions $(q^\star,v^\star,u^\star)$ exhibit an \textit{hyperbolic velocity turnpike} w.r.t.\ $(\bar{v},\bar{u}) = (0,0)$. %
	Here, the constant~$C$ in Inequality~\eqref{eq:vTP} can then be chosen analogously to \cite{kit:faulwasser19b_2} %
	with a slight correction in order to account for the additional summand representing the influence of the incoming and leaving arc. %
	This term also necessitates the restriction of the time domain using an appropriately chosen constant~$\bar\nu$. %
	The additional assertion w.r.t.\ the adjoint variable~$\lambda_v$ directly follows from Equation~\eqref{eq:exampeLambdaV}. %
	Then, using the definition of $\cosh(T)$ and $\sinh(T)$ yields
	\begin{equation}
		T | \lambda_q(t) | = T | \lambda_q(0) | = \frac{T \sinh(T)}{(T-2) \sinh(T) + 2 (1-e^{-T})} | \eqref{eq:exampleFactor} |. \nonumber
	\end{equation}
	Then, using a series expansion analogously to~\cite{kit:faulwasser19b_2} for the fraction yields a term %
	which is uniformly upper bounded by $3/2$ if the first two summands $T^2 + T^4/6$ of the series expansion for $T \sinh(T)$ are neglected. But this summand, i.e.\ 
	\begin{equation}
		\frac {T^2 + T^4/6}{\sum_{k=2}^\infty \frac {T^{2k}}{(2k-1)!} (1-\frac 1k) },	\nonumber
	\end{equation}
	is rapidly decaying to zero for sufficiently large~$T$, which shows the assertion for appropriately chosen~$\bar\nu$.
\end{proof}
	
	Theorem~\ref{ThmMotivationalExample} extends \cite[Proposition~8]{kit:faulwasser19b_2} to non-zero initial and terminal velocity and %
	explains the respective incoming and leaving arcs. Moreover, it also covers the behaviour of the adjoint variables.

\begin{rem}[Relation to the velocity turnpike property] %
	Let $\nu: [0,\infty) \rightarrow [0,\infty]$ be defined by
	\begin{equation}
		\varepsilon \mapsto \begin{cases}
			\infty & \varepsilon = 0 \\
			\max\{2 \bar{\nu}, C/\varepsilon\} & \varepsilon > 0
		\end{cases} \nonumber
	\end{equation}
	with $\bar\nu$ and $C$ from Theorem~\ref{ThmMotivationalExample}.	Then, $(\bar{v},\bar{u})^\top$ satisfies Definition~\ref{def:turnpike} since %
	either the horizon length~$T$ is \textit{sufficiently long}, i.e.\ $C/T \leq \varepsilon$ holds. %
	Then, only the incoming and the leaving arc may violate the desired inequality resulting in $2 \nu$. %
	Otherwise, the horizon~$T$ is smaller than $C/\varepsilon$ such that the inequality trivially holds. %
	In conclusion, this reasoning shows that a bound like the one derived in Theorem~\ref{ThmMotivationalExample} always implies the (measure-based) velocity turnpike property.\eBox
\end{rem}
Optimal solutions for an example scenario, namely $x_0 = 0.0$, $x_T = 5.0$, $v_0 = 3.0$, $v_T = 6.0$, are shown in Figure~\ref{fig:motivational Example}. We give the two state and two adjoint variables for $T \in \{5,10,15,20,25,30,35\}$.
\begin{figure}
	\begin{center}
		\includegraphics[width=.75\textwidth]{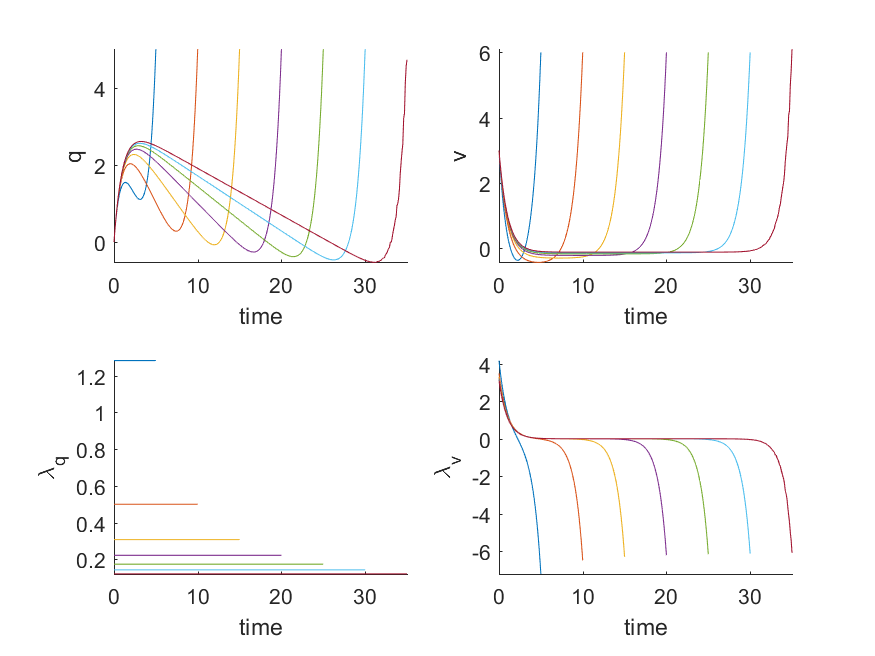}
		\caption{Numerical solution of the illustrative example  for $T = 20$.}
		\label{fig:motivational Example}
	\end{center}
\end{figure}
Here, the optimal solution has the predicted turnpike property at $\bar v = \bar u =0$ with zero control and %
thus constant velocity and linear decrease of configuration.
The incoming and leaving arc ensure that the boundary conditions on the configuration and velocity components are met.

\section{Conclusions} \label{sec:Conclusion}
This paper has investigated the relation between dissipativity properties of OCPs and velocity turnpikes. We extended our previous results from  \cite{kit:faulwasser19b_2} by adding a sufficient condition based on dissipativity and by making explicit the link between optimal trim solutions, which correspond to velocity steady states, and the turnpike. 
To this end, we considered a special type of symmetry, namely the invariance of the dynamics w.r.t. to the full configuration vector $q$. This simplifies the characterization of trims to defining tuples $(v,u)$, i.e.\ trims are defined by their constant velocity (and $u$ is chosen to satisfy $f(v,u)=0$).
Future work will explore more general symmetry properties and converse turnpike results. 

\bibliographystyle{abbrv}  
\footnotesize
\bibliography{ReferencesTurnpikeMechSystem}             
                                                   







\end{document}